\documentclass[
    a4paper,
    USenglish,
    cleveref,
    autoref,
    thm-restate,
]{lipics-v2021}

\nolinenumbers
\hideLIPIcs

\usepackage{algorithm}
\usepackage{algpseudocode}
\usepackage{cases}
\usepackage{mathtools}
\usepackage{amsmath}
\usepackage[multiple]{footmisc}

\usepackage{subcaption}
\usepackage{tabularx}
\usepackage{booktabs} 

\usepackage[textsize=tiny]{todonotes}

\title{Did Fourier Really Meet M\"obius? Fast Subset Convolution via FFT} 

\titlerunning{Did M\"obius Really Meet Fourier? Fast Subset Convolution via FFT} 

\author{Mihail Stoian}{University of Technology Nuremberg, Germany\and \url{https://stoianmihail.github.io} }{mihail.stoian@utn.de}{https://orcid.org/0000-0002-8843-3374}{}

\authorrunning{M. Stoian} 

\Copyright{Mihail Stoian} 



\category{} 

\relatedversion{} 

\acknowledgements{}

\usepackage{amsthm}
\newcommand{\sparagraph}[1]{\vspace{1mm}\noindent {\bf #1}}

\makeatletter
\newcommand{\ovee}{\mathbin{\mathpalette\make@circled\vee}}
\newcommand{\make@circled}[2]{%
  \ooalign{$\m@th#1\smallbigcirc{#1}$\cr\hidewidth$\m@th#1#2$\hidewidth\cr}%
}
\newcommand{\smallbigcirc}[1]{%
  \vcenter{\hbox{\scalebox{0.82}{$\m@th#1\bigcirc$}}}%
}

\makeatletter
\newcommand{\subalign}[1]{%
  \vcenter{%
    \Let@ \restore@math@cr \default@tag
    \baselineskip\fontdimen10 \scriptfont\tw@
    \advance\baselineskip\fontdimen12 \scriptfont\tw@
    \lineskip\thr@@\fontdimen8 \scriptfont\thr@@
    \lineskiplimit\lineskip
    \ialign{\hfil$\m@th\scriptstyle##$&$\m@th\scriptstyle{}##$\hfil\crcr
      #1\crcr
    }%
  }%
}
\makeatother

\begin{document}

\maketitle

\begin{abstract}
In their seminal work on subset convolution, Bj\"orklund, Husfeldt, Kaski and Koivisto introduced the now well-known $O(2^n n^2)$-time evaluation of the subset convolution in the sum-product ring. This sparked a wave of remarkable results for fundamental problems, such as the minimum Steiner tree and the chromatic number. However, in spite of its theoretical improvement, large intermediate outputs and floating-point precision errors due to alternating addition and subtraction in its set function transforms make the algorithm unusable in practice.

We provide a simple FFT-based algorithm that completely eliminates the need for set function transforms and maintains the running time of the original algorithm. This makes it possible to take advantage of nearly sixty years of research on efficient FFT implementations.
\end{abstract}

\section{Introduction}\label{sec:introduction}

Subset convolution is one of the few examples of tools that need no introduction -- at least in exact or parameterized algorithms. As a sign for its importance, it even made it into the prominent book on Parameterized Algorithms~\cite{Cygan2015_chapter10}. Its popularity is also due to the suite of applications in which is used, such as the minimum Steiner tree, the chromatic number, and many more, even outside theoretical computer science (see Ref.~\cite{fsc} for further applications).

Namely, given $f$ and $g$ two set functions defined on the subset lattice of order $n$, their subset convolution on the sum-product ring is defined as
\begin{equation}
    (f \ast g)(S) = \displaystyle\sum_{T \subseteq S} f(T)g(S \setminus T),
    \label{eq:subset_conv}
\end{equation}
for any $S \subseteq [n]$. The na\"ive evaluation takes $O(3^n)$-time as for each cardinality we need to iterate all the subsets of that cardinality. In their seminal work, Bj\"orklund et al.~\cite{fsc} showed that there is a faster way: They introduced an $O(2^n n^2)$-time algorithm, using the zeta and M\"obius transform, a faster algorithm for these  being already due to Yates~\cite{yates_algo}.

\sparagraph{Motivation.} As it stands now, the algorithm is far from being used smoothly in practice. For instance, Kohonen and Corander~\cite{kohonen2016computing} refrained from using it in their application, and instead used the na\"ive evaluation. Their reason was rounding errors in floating-point arithmetic, due to alternating additions and subtractions in the M\"obius transform. Another problem we encountered ourselves was the magnitude of the values in the intermediate outputs, which are themselves zeta transforms of the final convolution. In contrast, the well-known $O(n \log n)$-time algorithm of the fast Fourier transform (FFT)~\cite{cooley1965algorithm} is well-understood and has witnessed a rich line of research on practical implementations, e.g., the FFTW library~\cite{frigo1999fast}; see Ref.~\cite{fft_survey_1} for a recent survey.

\sparagraph{Contribution.} We provide a simple algorithm for the subset convolution in the $(+, \times)$-ring that does not rely on set function transforms. Notably, our FFT-based algorithm maintains the running time of $O(2^n n^2)$ of the original algorithm proposed by Bj\"orklund et al.~\cite{fsc}. This makes it possible to take advantage of nearly sixty years of research on efficient FFT implementations.

\section{Algorithm}

Most importantly, our algorithm eliminates the need for set function transforms. Instead, we interpret the set functions as vectors indexed by the natural binary representation of subsets, and transform them via FFT. We outline the pseudocode in Alg.~\ref{alg:subset_conv_via_fft}.

\begin{algorithm}
    \caption{$\textsc{SubsetConvolutionViaFFT}(f,g)$}
	\label{alg:subset_conv_via_fft}
\begin{algorithmic}[1]
    \State $\{\hat f^{(0)}, \ldots, \hat f^{(n)}\} \gets \textsc{FFT}(\textsc{Chop}(f))$
    \State $\{\hat g^{(0)}, \ldots, \hat g^{(n)}\} \gets \textsc{FFT}(\textsc{Chop}(g))$
    \State $h^{(k)} \gets \textsc{IFFT}(\sum_{i=0}^{k} \hat f^{(i)}\hat g^{(k - i)})$, for any $0 \leq k \leq n$.
    \State $h^{(k)}[\{S \subseteq [n] \mid |S| \neq k\}] \gets 0$, for any $0 \leq k \leq n$.
    \State \Return $\sum_{k=0}^{n} h^{(k)}$
\end{algorithmic}
\end{algorithm}

Let us analyze the algorithm. First, we ``chop'' the set functions into $n + 1$ set functions and transform them via FFT.\footnote{In the pseudocode, $\textsc{FFT}(\cdot)$ maps the ``chopped'' functions to their FFT counterpart.} A chopped function $f^{(i)}$ is a copy of the original function $f$ that takes values only for sets $S \subseteq [n]$ with $|S| = i$; the rest is set to 0. This is similar to  ranked set functions in Bj\"orklund et al.~\cite{fsc}. Most importantly, we do \emph{not} transform these via the zeta transform, but rather directly via FFT. Next, we fix a cardinality $0 \leq k \leq n$ and compute the convolution in the frequency domain. This step resembles the ranked convolution in Ref.~\cite{fsc}. We can then apply the inverse FFT on this output. At this step, the previous algorithm would have applied the M\"obius transform.\footnote{The zeta and M\"obius transform are the inverse of each other.} At this point, the ``vanilla'' $h^{(k)}$'s in line~3 may have values for sets of cardinality other than $k$. These need to be set to~0 (line~4) before we assemble the final output (line 5).
\begin{theorem}
    Algorithm~\ref{alg:subset_conv_via_fft} correctly computes the subset convolution in the $(+, \times)$-ring.
\end{theorem}
\begin{proof}
    The fine work goes into line 3 of the algorithm. First, note that this step is equivalent to computing $\sum_{i=0}^{k} \textsc{IFFT}(\hat f^{(i)}, \hat g^{(k - i)})$, due to the linearity of the (inverse) Fourier transform. Since we remove the contribution of sets with cardinality different $k$ in the next line, we only need to analyze the correctness of the values of those set $S \subseteq [n]$ with $|S| = k$. To this end, we use the following equivalence for two integers $a, b \in [2^n]$:\[
        a + b = (a\:|\:b) + (a\:\&\:b),
    \]
    where ``$|$'' and ``$\&$'' are the well-known bit operations, which have the natural interpretation in the set operations ``$\cup$'' and ``$\cap$'', respectively. Thus, multiplying $\hat f^{(i)}$ and $\hat g^{(k-i)}$ enforces that the sets to be multiplied are \emph{disjoint} since their cardinalities sum up to $k$. In other words, their ``$\&$''-result is 0, i.e., the empty set. Summing over all possible cardinalities $0 \leq i \leq k$ ensures that all subsets are considered, analogous to Bj\"orklund et al.~\cite{fsc}.
\end{proof}

\sparagraph{Running Time.} The chopped functions can be computed in $O(2^n n)$-time. Applying FFT on these takes time $O(n \cdot 2^n \log 2^n) = O(2^n n^2)$. For any $0 \leq k \leq n$, the evaluation of $h^{(k)}$ can be done in time $O(2^n n + 2^n \log 2^n) = O(2^n n)$, hence $O(2^n n^2)$ for this step. Finally, ``cleaning up'' the $h^{(k)}$'s and summing them runs in $O(2^n n)$-time. The total running time thus reads $O(2^n n^2)$, maintaining the running time of the original algorithm~\cite{fsc}. On a practical note, the input does not need any zero-padding as the size itself is always a power of two.

\section{Discussion}

\sparagraph{Implications.}~While our work shows that subset convolution does not require any specialized set function transforms, this does not make them redundant: They are relevant by themselves in order theory~\cite{bjorklund2015fast, chaveroche2021focal, kaski2016fast, pegolotti2022fast}, Bayesian networks~\cite{parviainen2010bayesian}, and in reinterpreting dynamic programming recursions of well-known problems to obtain polynomial-space algorithms~\cite{nederlof2009fast}. Beyond simplicity, our work enables the use of efficient FFT libraries. This is particularly relevant for practical implementations of dynamic programs on tree decompositions~\cite{van2009dynamic}, which are widely used in algorithm engineering; see, for instance, the PACE 2018 challenge~\cite{pace_2018}.

\sparagraph{Related Work.}~To ensure that our algorithm has not appeared as a by-product of a previous work, we have examined the articles citing Bj\"orklund et al.~\cite{fsc} that contain the terms ``Fourier'' and ``FFT'' in their text.\footnote{This is easily done using Google Scholar's feature of searching for a pattern within citing articles.} In particular, Cygan and Pilipczuk's work on exact and approximate bandwidth~\cite{cygan2010exact} comes the closest: They solve the \textsc{Disjoint Set Sum} problem using both the standard fast subset convolution algorithm and a custom FFT solution (see their Sec.~2.2). While they do not solve the fast subset convolution via FFT, their algorithm for this particular problem bears similarity with ours. Interestingly enough, at a closer inspection, we can even improve their $O(2^n n^3)$-time algorithm for the \textsc{Disjoint Set Sum} problem by a linear factor, since they do not exploit the linearity of the (inverse) Fourier transform and simply perform $n^2$ sequence convolutions.

On a final note, our work further ``sinks'' the recently uncovered algorithmic isthmus between sequence convolution and subset convolution~\cite{stoian2024sinking}.
\bibliography{simple_fsc}

\end{document}